\documentclass[11pt]{article}
\usepackage{amssymb}
\usepackage[all]{xy}
\usepackage{graphicx}

\title{The Logos Categorical Approach to QM: II.\\Quantum Superpositions.}

\author{{\sc C. de Ronde}$^{1,2}$ and {\sc C. Massri}$^{3,4}$}
\date{}

\usepackage[margin=2.5cm]{geometry}

\begin{document}

\bibliographystyle{plain}
\maketitle

\begin{center}
\begin{small}
1. Philosophy Institute Dr. A. Korn (UBA-CONICET)\\
2. Center Leo Apostel for Interdisciplinary Studies\\Foundations of the Exact Sciences (Vrije Universiteit Brussel)\\
3. Institute of Mathematical Investigations Luis A. Santal\'o (UBA-CONICET)\\
4. University CAECE
\end{small}
\end{center}

\begin{abstract}
\noindent In this paper we attempt to consider quantum superpositions from the perspective of the logos categorical approach presented in \cite{deRondeMassri17a}. We will argue that our approach allows us not only to better visualize the structural features of quantum superpositions providing an {\it anschaulich} content to all terms, but also to restore ---through the intensive valuation of graphs and the notion of {\it immanent power}--- an objective representation of what QM is really talking about. In particular, we will discuss how superpositions relate to some of the main features of the theory of quanta, namely, contextuality, paraconsistency, probability and measurement. 
\end{abstract}
\begin{small}

\end{small}

\newtheorem{theo}{Theorem}[section]
\newtheorem{definition}[theo]{Definition}
\newtheorem{lem}[theo]{Lemma}
\newtheorem{met}[theo]{Method}
\newtheorem{prop}[theo]{Proposition}
\newtheorem{coro}[theo]{Corollary}
\newtheorem{exam}[theo]{Example}
\newtheorem{rema}[theo]{Remark}{\hspace*{4mm}}
\newtheorem{example}[theo]{Example}
\newcommand{\proof}{\noindent {\em Proof:\/}{\hspace*{4mm}}}
\newcommand{\qed}{\hfill$\Box$}
\newcommand{\ninv}{\mathord{\sim}} 
\newtheorem{postulate}[theo]{Postulate}

\bigskip

\bigskip

\bigskip

\bigskip

\bigskip

\bigskip

\bigskip

\bigskip

\bigskip

\bigskip

\bigskip

\section*{Introduction}

Quantum superpositions are certainly one of the most important elements within the new quantum technological era that is taking place today in quantum information processing. Unfortunately, the many technological and formal advances in quantum computation, teleportation and cryptography are not being accompanied by the community interested in philosophical and foundational issues about Quantum Mechanics (QM). As one of us has argued elsewhere \cite{deRonde17a}, instead of analyzing and learning about the physical meaning of quantum superpositions, the orthodox project in philosophy of QM has concentrated its efforts in trying to ``bridge the gap'' between the ``weird'' mathematical formalism and our classical ``manifest image of the world''. This project is grounded in two Bohrian (metaphysical) presuppositions ---turned today into dogma--- according to which, first, we must {\it necessarily} presuppose the existence of a reductionistic ``limit'' between QM and classical physics (i.e., the famous quantum to classical limit), and second, we must accept that the language used by physicists in order to account for observations {\it will be forever} constrained by that of Newton and Maxwell. 

Taking distance from this ``conservative project'' which grounds itself in the classical (metaphysical) worldview, our logos approach to QM ---as presented in \cite{deRondeMassri17a}--- advances exactly in the opposite direction. Namely, our proposal seeks to understand the theory completely independently of classical physics and its metaphysical representation.\footnote{As we discussed in detail in [{\it Op. cit.}], the classical representation amounts to an understanding of physical reality in terms of an {\it actual state of affairs}; more specifically, in terms of `systems' constituted by definite valued `properties'.} Furthermore, when attempting to develop such a (non-classical) representation in order to understand QM, we do not believe to be {\it necessarily} constrained by the conceptual schemes created less than four centuries ago by Newton and Maxwell. In fact, we see the history of physics, philosophy and mathematics as one of continuos development and creation. A creation that is not linear nor reductionistic but, on the very contrary, one that is multiple and plural. It is exactly this ontological pluralist viewpoint (see \cite{deRonde16b}) which allows us to avoid the reductionistic questioning imposed by what could be called ``the Bohrian orthodoxy''. 

In the present paper we attempt to discuss, continuing the analysis presented in \cite{deRonde17a}, the physical meaning of quantum superpositions. For this purpose, we attempt to derive from our logos categorical formalism ---which makes explicit use of the notions of {\it immanent power} and {\it potentia}--- a new representation that is capable of providing an {\it anschaulich} content to the theory. That is, an intuitive approach that is able to explain and understand, in both formal and conceptual manner, what the theory of quanta is really taking about. 

The paper is organized as follows. In section 1 we discuss the superposition problem which attempts to change the perspective regarding quantum superpositions commonly addressed in terms of the infamous measurement problem. In section 2 we discuss a possible solution to the superposition problem in terms of the notion of {\it immanent intensive power}. Section 3 presents an explicit definition of `quantum superposition' in the logos formalism. After this definition, we discus the visualization and understanding of some important features of superpoitions such as contextuality, paraconsistency, probability and measurement. Finally, we end the paper with some general conclusions.

\section{The Superposition Problem} 

Very unfortunately, quantum superpositions are regarded today as unwanted guests in the philosophical literature about QM. Their characterization is most often accompanied by negative adjectives like ``embarrassing'', ``weird'', ``strange'', ``ghostly'', ``spooky'', etc. In a recent interview on the subject, superpositions were even pictured as a ``contagious disease'' (See \cite{Wallace17}). The unease of philosophers of physics with quantum superpositions seems to come from the fact that their existence apparently contradicts what we actually observe in the lab. Indeed, we never observe a pointer reading in a superposed state of `1' and `4' ---as predicted by QM. This has lead many researchers to believe that ---using the words of Robert Griffiths \cite{Griffiths13}--- we should better get ``rid of the ghost of Schr\"odinger's cat.'' 

We believe that this gloomy attitude found in the portrayal of quantum superpositions is a consequence of the deep widespread ---sometimes explicit, but most of the time implicit--- commitment to the Bohrian doctrine of concepts according to which our representation of ``classical reality'' must be regarded not only as a standpoint for a proper understanding of {\it what is observed}, but also as an ending-point regarding the conceptual representation of {\it what is really going on}. It is this same (metaphysical) prejudice regarding our limited possibilities to represent reality which has lead the philosophical community interested in QM to produce an analysis of the quantum formalism in purely ``negative terms''. That is, in terms of the failure of the quantum formalism to account for the classical representation of reality and the use of its concepts: separability, space, time, locality, individuality, identity, actuality, etc. As a consequence, the ``negative problems'' that we find commonly discussed in the literature are: {\it non-}separability, {\it non-}locality, {\it non-}individuality, {\it non-}identity, etc. These ``no-problems'' begin their analysis considering the notions of classical physics, taking for granted the very strong metaphysical presupposition according to which QM should be able to represent physical reality using exactly these same notions. This has created a paradoxical situation in the field. The commitment to our classical representation of the world is pursued by most lines of research even though everything we have learned since the very birth of QM seems to point exactly in the opposite direction. That is, that QM cannot be understood in terms of classical notions. 

The infamous Measurement Problem (MP) makes even more explicit the orthodox program with respect to quantum superpositions. The focus of the MP is that we never perceive macroscopic superpositions (e.g., \cite{HawkingPenrose, Wallace17}). Presupposing the idea that measurement outcomes are unproblematic, unquestionable {\it givens} of observation, the attention is then focused in trying to justify the observed (classical) measurement outcomes ---which are never observed as ``superposed pointer readings''. \\ 

\noindent {\it {\bf Measurement Problem:} Given a specific basis (or context), QM describes mathematically a quantum state in terms of a superposition of, in general, multiple states. Since the evolution described by QM allows us to predict that the quantum system will get entangled with the apparatus and thus its pointer positions will also become a superposition,\footnote{Given a quantum system represented by a superposition of more than one term, $\sum c_i | \alpha_i \rangle$, when in contact with an apparatus ready to measure, $|R_0 \rangle$, QM predicts that system and apparatus will become ``entangled'' in such a way that the final `system + apparatus' will be described by  $\sum c_i | \alpha_i \rangle  |R_i \rangle$. Thus, as a consequence of the quantum evolution, the pointers have also become ---like the original quantum system--- a superposition of pointers $\sum c_i |R_i \rangle$. This is why the measurement problem can be stated as a problem only in the case the original quantum state is described by a superposition of more than one term.} the question is why do we observe a single outcome instead of a superposition of them?}\\

\noindent As obvious as it might seem, this perspective grounded on ---what are considered to be--- ``classical clicks in measurement apparatus'' ---or, in more general terms, ``classical observations''--- leaves implicitly aside the analysis regarding the conceptual meaning and representation of quantum superpositions themselves. As we just mentioned, this analysis is implicitly grounded on Bohr's doctrine of classical concepts according to which \cite[p. 7]{WZ}: ``[...] the unambiguous interpretation  of any measurement must be essentially framed in terms of classical physical theories, and we may say that in this sense the language of Newton and Maxwell will remain the language of physicists for all time.'' And in this respect,  ``it would be a misconception to believe that the difficulties of the atomic theory may be evaded by eventually replacing the concepts of classical physics by new conceptual forms.'' From this viewpoint the possibilities might seem to be not so many. If we accept that we can only workout the understanding of QM through classical notions, and that we do not observe superpositions in the macro-world, then one could argue that the observation of the other terms is taking place ``elsewhere'', maybe in some other different ``world'' or ``universe''. David Deutsch \cite{Deutsch97} has argued that the only way of understanding quantum information processing and the exponential speed-up of quantum computation is by accepting the ``parallel'' existence of many worlds.\footnote{For an insightful critique of the idea that quantum computation implies the existence of many worlds, see \cite{Steane03}.} Following this line of reasoning, David Wallace argues that ``the existence of a multiverse'' can be read out ``literally'' from out the formalism of QM itself.\footnote{According to Wallace [{\it Op. cit.}] there are only three possible choices when attempting to address the interpretational problems of quantum theory. The first possibility is to give up on philosophy and accept instrumentalism, the second is to give up on physics and provide a new explanation of why superpositions are not being observed, the third and last possibility is to believe in the existence of many worlds.} There is no other option, to take the formalism ``seriously'' means to ``believe in the existence of many worlds'' \cite{Wallace12, Wallace17}. Roger Penrose, pointing to the other end of the dilemma ---namely, the observer---, argues instead that the solution to the riddle might come from the analysis of our own conscious human perception. The problem might be solved if we understand why we are not able to observe these strange superpositions \cite{Penrose17}. Both solutions are very difficult to swallow for a realist. The first seems to go against one of the most basic realist beliefs in what she observes, and that is of course, that we live in only one world. No-one has ever observed these ---supposedly, also existent--- many worlds. Nor there is any indication that these ``new worlds'' point directly to new experimental possibilities, different from those already contained in the orthodox account of QM. So adding a multiplicity of {\it unobservable worlds} in order to explain the existence of a pointer which was not observed (at the end of a quantum measurement) seems to be extending, rather than explaining, our incapability of observation. While with the MP we had accepted to some extent our failure in observing a (superposed) pointer reading, with the ``many-worlds solution'' it seems we must accept our almost complete incapability to observe most of what is going on in reality ---i.e., the unobservable multiverse. The second solution proposed by Penrose ---going back to Wigner's analysis of QM \cite{WZ}---, seems to be even worse for any realist who takes seriously the first premise of his own program according to which the description of reality must be accounted for independently of any conscious being. Indeed, claiming that understanding  consciousness might allow us to explain (quantum) physical reality seems to be going against the very basic presupposition of realism itself.  

All we have just said, about many worlds and consciousness as the only possible solutions to the MP, stands only when taking for granted the Bohrian doctrine of classical concepts. However, once we escape the gates of classical representation the possibilities to produce a coherent physical account of quantum superpositions change drastically. Once we give up on classical concepts we are not anymore committed to their explanation, either by creating unobservable worlds or by shifting our attention to the problem of consciousness. If we do what Bohr had explicitly forbid us, and create new conceptual forms, maybe there is a different way out. Maybe we need to think differently, maybe we need to ask different questions.

In \cite{deRonde17a}, one of the authors of this paper proposed a new focus in the questioning regarding the problem of quantum superpositions:\\

\noindent {\it {\bf Superposition Problem:} Given a situation in which there is a quantum superposition of more than one term, $\sum c_i | \alpha_i \rangle$, and given the fact that each one of the terms relates through the Born rule to a Meaningful Operational Statements (MOS),\footnote{Recalling the definition provided in \cite{deRonde17a}, a {\emph Meaningful Operational Statements} can be defined in the following manner: {\it Every operational statement within a theory capable of predicting the outcomes of possible measurements must be considered as meaningful with respect to the representation of physical reality provided by that theory.}} the question is how do we conceptually represent this mathematical expression? Which is the physical notion that relates to each one of the terms in a quantum superposition?}\\
 
\noindent The important developments we are witnessing today in quantum information processing demands us, philosophers of QM, to pay special attention to the novel requirements of this new technological era. The superposition problem opens the possibility to truly discuss a conceptual representation of reality which goes beyond the Newtonian metaphysical representation of systems with definite valued properties. According to this new problem, we do not need to focus in what is not observed according to classical physics; instead, we need to understand observation itself from a completely different angle.

\section{A Conceptual Representation of Quantum Superpositions\\ (Beyond the Actual Realm)} 

Both the measurement and the superposition problems imply the necessary requirement that the quantum superposition is formally defined in a contextual manner. The fact that a ray in Hilbert space, $v_{\Psi}$, describing a state of affairs can be mathematically represented in multiple bases must be explicitly considered within such definition. This obvious remark might be regarded as controversial due to the fact the contextual character of quantum superpositions has been completely overlooked within the orthodox literature. The obvious distinction between two different levels of mathematical representation has been dissolved by the (classical) semantical interpretation of the quantum formalism in terms of `systems', `states' and `properties'. Indeed, as we have discussed in detail in \cite{daCostadeRonde16}, the use of these notions has camouflaged the important distinction between an abstract vector, $v_{\Psi}$, and its different basis-dependent representations, $\sum c_i | \alpha_i \rangle$.\footnote{The semantic interpretation used in order to interpret the syntactical level of the quantum formalism presupposes implicitly PE, PNC and PI. This ``common sense''  classical interpretation has been uncritically accepted without considering the required necessary coherency between the addressed semantical and syntactical levels of the theory.} In order to make explicit these two different formal levels, we have provided in \cite{deRonde17a} the following contextual definition of quantum superpositions:\\

\noindent {\it {\sc Quantum Superposition:} Given a quantum state, $v_{\Psi}$, each i-basis defines a particular mathematical representation of $v_{\Psi}$, $\sum c_i | \alpha_i \rangle$, which we call a quantum superposition. Each one of these different basis-dependent quantum superpositions defines a specific set of \emph{Meaningful Operational Statements.} These statements are related to each one of the terms of the particular quantum superposition through the Born rule. The notion of quantum superposition is contextual for it is always defined in terms of a particular experimental context (or basis).}\\

\noindent As we shall see in the next section, this definition finds a visualizable expression in the logos approach which assumes the particular conceptual representation grounded on the notion of {\it immanent intensive power} (see \cite{deRonde17a, deRonde17c, deRonde17d, deRonde17b, deRondeMassri17a}). Before discussing how the logos approach is able to provide an intuitive account of quantum superpositions let us discuss the main ideas behind this conceptual representation. 

A specific vector $v_{\Psi}$ with no specified mathematical basis in Hilbert space represents a {\it Potential State of Affairs} (PSA); i.e., the definite potential existence of an aggregate of {\it immanent powers}, each one of them with a specific {\it potentia} or {\it intensity}. Given a PSA, $\Psi$, and the context or basis, we call a {\it quantum situation} any superposition of one or more than one power. In general, given the basis $B= \{ | \alpha_i \rangle \}$ the quantum situation $QS_{\Psi, B}$ is represented by the following superposition of immanent powers:
\begin{equation}
c_{1} | \alpha_{1} \rangle + c_{2} | \alpha_{2} \rangle + ... + c_{n} | \alpha_{n} \rangle
\end{equation}

\noindent We write the quantum situation of the PSA, $\Psi$, in the context $B$ in terms of the ordered pair given by the elements of the basis and the coordinates of the vector $v_{\Psi}$ of the PSA in that basis:

\begin{equation}
QS_{\Psi, B} = (\Psi, B)
\end{equation}

\noindent The elements of the basis, $| \alpha_{i} \rangle$, are interpreted in terms of {\it immanent powers}. The coordinates of the elements of the basis in square modulus, $|c_{i}|^2$, are interpreted as the {\it potentia} of the power $| \alpha_{i} \rangle$, respectively. Given the PSA and the context, the quantum situation, $QS_{\Psi, B}$, is univocally determined in terms of a set of powers and their respective coordinates. Let us remark that in contradistinction with the notion of {\it quantum state} ---which allows us to refer to {\it the same} quantum system irrespectively of the basis---, the definition of a {\it quantum situation} ---which discusses the aggregation of a specific group of powers--- is basis dependent and thus intrinsically contextual. In this way, the notion of quantum situation captures in a natural manner the contextual character of quantum measurements without braking down the objective nature of the formal representation.  

In QM we only observe, within the actual realm, discrete shifts of energy (quantum postulate). These discrete shifts are interpreted in terms of {\it elementary processes} which produce actual effectuations. An elementary process is the path which undertakes a power from the potential realm to its actual effectuation. This path is governed by the {\it immanent cause} which allows the power to remain potentially preexistent within the potential realm independently of its actual effectuation.\footnote{For a detailed exposition of how the measurement process is explained in terms of the notion of {\it immanent cause} see \cite{deRonde17d}.} Each power $| \alpha_{i} \rangle$ is univocally related to an elementary process represented by the projection operator $P_{\alpha_{i}} = | \alpha_{i} \rangle \langle \alpha_{i} |$. Immanent powers exist in the mode of being of ontological potentiality. An {\it actual effectuation} is the expression of a specific power within actuality. Different actual effectuations expose the different powers of a given $QS$. In order to learn about a specific PSA (constituted by an aggregate of powers and their potentia) we must measure repeatedly the actual effectuations of each power exposed in the laboratory. A laboratory is understood as the set of all possible experimental arrangements, or quantum situations, related to a particular $\Psi$. Once again, we remark that actual effectuations do not change or affect in any way the PSA. Actual effectuations provide a glimpse of the PSA in the same way that `looking at table from above' gives us only a partial profile of the object; i.e. the ASA.

A {\it potentia} quantifies the intensity of an immanent power which exists (in ontological terms) in the potential realm; it also provides a measure of the possibility to express itself (in epistemic terms) in the actual realm. Given a PSA, the potentia is represented via the Born rule. The potentia $p_{i}$ of the immanent power $| \alpha_{i} \rangle$, in the specific PSA, $\Psi$, is given by:
\begin{equation}
p_{(| \alpha_{i} \rangle, \Psi)} = \langle \Psi | P_{\alpha_{i}} | \Psi \rangle = Tr[P_{ \Psi} P_{\alpha_{i}}]
\end{equation}

In order to learn about a $QS$ we must observe not only its powers (which are exposed in actuality through actual effectuations) but we must also measure the potentia of each respective power. In order to measure the potentia of each power we need to expose the $QS$ statistically through a repeated series of observations. The potentia, given by the Born rule, coincides with the probability frequency of repeated measurements when the number of observations goes to infinity. This representation, which extends the consideration of what can be considered to be an element of physical reality, allows us to claim that the probability found via the Born rule provides objective knowledge about the PSA ---rather than ignorance about an {\it actual state of affairs} \cite{deRonde16a, deRondeMassri17a}.\\

To conclude, according to this representation, QM talks mainly about the evolution and interaction of immanent powers and their respective potentia. This new conceptual scheme provides an intuitive grasp of how to think about QM, not only within the only world we have ever observed, but also without making any reference to conscious beings. The examples discussed in \cite{deRonde16a, deRondeMassri17a} present an objective counterfactual discourse which might allow us, not only to account for what is going on according to the theory of quanta, but also to characterize what we observe according to it. In the following section we will see how this conceptual representation finds a clear visualizable exposition within the logos categorical formalism.

\section{Quantum Superpositions in the Logos Approach}

Following \cite{deRondeMassri17a}, in this section we provide  some basic notions regarding our logos categorical approach. We assume that the reader is familiar with the definition of a \emph{category}. 

Let $\mathcal{C}$ be a category and let $C$ be an object in $\mathcal{C}$. Let us define the category over $C$ denoted $\mathcal{C}|_C$.
Objects in $\mathcal{C}|_C$ are given by arrows to $C$, $p:X\rightarrow C$,  $q:Y\rightarrow C$, etc. Arrows $f:p\rightarrow q$
are commutative triangles,
\[
\xymatrix{
X\ar[rr]^f\ar[dr]_p& &Y\ar[dl]^q\\
&C
}
\]

\noindent For example, let $\mathcal{S}ets|_\mathbf{2}$ be the category of sets
over $\mathbf{2}$, where $\textbf{2}=\{0,1\}$ and $\mathcal{S}ets$ is
the category of sets.
Objects in $\mathcal{S}ets|_\mathbf{2}$
are functions from a set to $\{0,1\}$
and morphisms are commuting triangles, 
\[
\xymatrix{
X\ar[rr]^f\ar[dr]_p& &Y\ar[dl]^q\\
&\{0,1\}
}
\]
In the previous triangle, $p$ and $q$ are objects of 
$\mathcal{S}ets|_\mathbf{2}$
and $f$ is a function satisfying $qf=p$.

\

Our main interest is the category $\mathcal{G}ph|_{[0,1]}$ of graphs over the interval $[0,1]$. Let us start reviewing some properties of the category of graphs. A \emph{graph} is a set with a reflexive symmetric relation. The category of graphs extends naturally the category of sets and the category of aggregates (objects with an equivalence relation). A set is a graph without edges. An {\it aggregate} is a graph  in which the relation is transitive. More generally, we can assign to a category a graph, where the objects are the nodes of the graph and there is an edge between $A$ and $B$ if $\hom(A,B)\neq\emptyset$.
Given that in a category we have a composition law, the resulting graph is an aggregate.

\begin{definition}
We say that a graph $\mathcal{G}$ is \emph{complete} if there is an edge between two arbitrary nodes. A \emph{context} is a complete subgraph (or aggregate) inside $\mathcal{G}$. A \emph{maximal context} is a context not contained properly in another context. If we do 
not indicate the opposite, when we refer to contexts we will be implying maximal contexts.
\end{definition}

\noindent For example, let $P_1,P_2$ be two elements of a graph $\mathcal{G}$. 
Then, $\{P_1, P_2\}$ is a contexts if $P_1$ is related to $P_2$, $P_1\sim P_2$. Saying differently, if there exists an edge between $P_1$ and $P_2$. In general, a collection of elements $\{P_i\}_{i\in I}\subseteq \mathcal{G}$ determine a {\it context} if $P_i\sim P_j$ for all $i,j\in I$. Equivalently, if the subgraph with nodes $\{P_i\}_{i\in I}$ is complete. 

\

Given a Hilbert space $\mathcal{H}$, 
we can define naturally a graph $\mathcal{G}=\mathcal{G}(\mathcal{H})$
as follows. Following [{\it Op. cit.}] the nodes are represented by {\it immanent powers} and there exists an edge between 
$P$ and $Q$ if and only if $[P,Q]=0$.
This relation makes $\mathcal{G}$ a graph (the relation is not transitive). We call this relation {\it quantum commuting relation}.\footnote{In QM the commuting relation is given by $[X_i,X_j]=0$, $[P_i,P_j]=0$, $[X_i,P_j]=i\hbar\delta_{ij}{\bf 1}$.}
Let us mention some properties of this
graph. For simplicity, assume that $\dim(\mathcal{H})=4$, hence the contexts
have 4 nodes (or immanent powers).\footnote{In general the dimension might be arbitrarily large. We choose here a small dimension to make the picture more easy.} We can picture the graph as
\begin{center}
\includegraphics[width=14em]{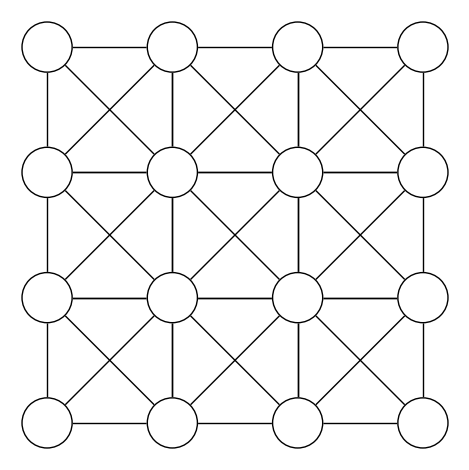}
\end{center}
Choosing a basis of $\mathcal{H}$ is equivalent to choosing a context,
\begin{center}
\includegraphics[width=14em]{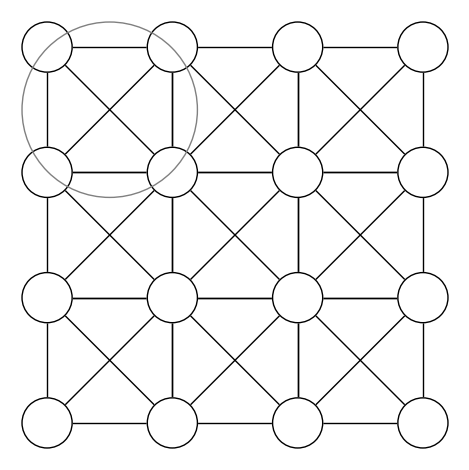}
\end{center}
But of course, within the same graph, we can also choose a different context,\footnote{In the orthodox interpretation the choice of the basis has been interpreted as ``the act of observing the quantum system''. This relation is not a direct one. Obviously, I can always choose a basis (on paper) without the need of performing any experiment (in the lab). Later on, the imposition of a binary valuation to the chosen basis is interpreted as something that ``actually takes place in reality''. The inference derived by Bohr in his reply to EPR \cite{Bohr35} is that the choice of the context determines the object under study. Today, in the words of Butterfiled \cite{Butterfield17}, the widespread conclusion is that: ``the properties of a system are different whether you look at them or not.'' For a detailed analysis we refer to \cite{deRonde16c, deRondeMassri17a}.}
\begin{center}
\includegraphics[width=14em]{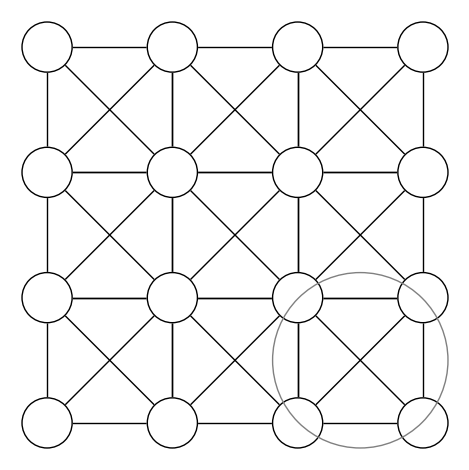}
\end{center}
A remarkable fact about this graph is that we can \emph{generate}
the whole graph starting with only one context and applying as many changes of basis as required.
For example, if we start with the following basis,
\begin{center}
\includegraphics[width=7em]{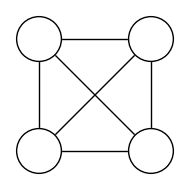}
\end{center}
we can apply a rotation (change of basis) to get a new context,
\begin{center}
\includegraphics[width=11em]{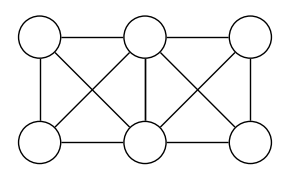}
\end{center}
Repeating this procedure we can then obtain the whole graph $\mathcal{G}$,
\begin{center}
\includegraphics[width=14em]{trescsg.png}
\end{center}
As a corollary of this procedure, we have the following:

\begin{coro}
Let $\mathcal{H}$ be a Hilbert space and let $\mathcal{G}$
be the graph of immanent powers with the commuting relation given by QM. 
It then follows that: 
\begin{enumerate}
\item The graph $\mathcal{G}$ contains all the contexts. 
\item Each context is capable of generating the whole graph $\mathcal{G}$.
\end{enumerate}
\end{coro}
\begin{proof}
Follows from the previous arguments.
\qed
\end{proof}

\

As we mentioned earlier, an object in $\mathcal{G}ph|_{[0,1]}$ consists in 
a map $\Psi:\mathcal{G}\rightarrow [0,1]$, where 
$\mathcal{G}$ is a graph. As before, to each node 
$P\in\mathcal{G}$, we assign through the Born rule the number $p=\Psi(P)$, but this time, 
$p$ is a number between $0$ and $1$.
The category $\mathcal{G}ph|_{[0,1]}$ has very nice categorical properties 
\cite{quasitopoi, graphtheory}, and is a \emph{logos}. As discussed in detail in \cite{deRondeMassri17a}, we call this 
map $\Psi:\mathcal{G}\rightarrow [0,1]$ a {\it Potential State of Affairs} (PSA for short)
and we picture it as follows,

\begin{center}
\includegraphics[width=14em]{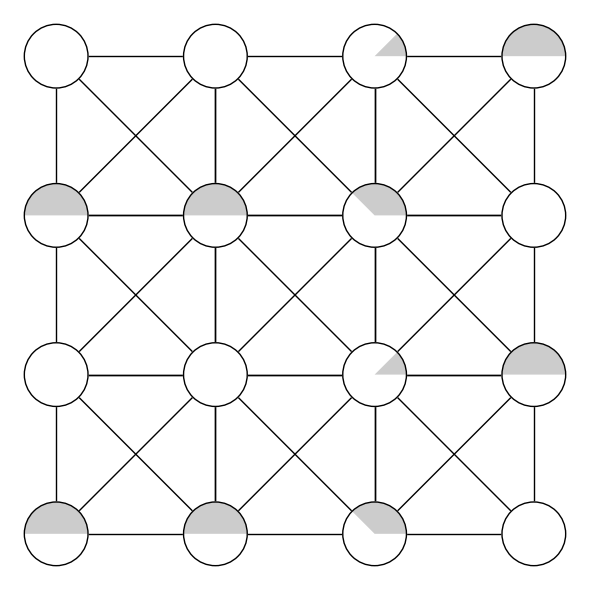}
\end{center}

\begin{definition}
Let $\mathcal{H}$ be Hilbert space and let $\rho$ be 
a positive semi-definite self-adjoint operator 
of the trace class (a density matrix).
Take $\mathcal{G}$ as the graph of immanent powers with the quantum commuting relation. 
To each immanent power $P\in\mathcal{G}$
apply the born rule to get the number $\Psi(P)\in[0,1]$, which is called the potentia (or intensity) of the power $P$. 
Then, $\Psi:\mathcal{G}\rightarrow [0,1]$
defines an object in $\mathcal{G}ph|_{[0,1]}$.
We call this map a \emph{Potential State of Affairs}.
\end{definition}

\noindent Notice that this mathematical representation is objective in the sense that it relates, in a coherent manner and without internal contradictions, the multiple contexts (or aggregates) to the whole PSA. Contrary to the contextual (relativist) Bohrian ``complementarity solution'', there is in this case no need of a (subjective) choice of a particular context in order to define the ``physically real'' state of affairs. The state of affairs is described completely by the whole graph (or $\Psi$), and the contexts bear an invariant existence independently of any choice. 

\begin{theo}
Let $\mathcal{H}$ be Hilbert space and 
let $\Psi:\mathcal{G}\rightarrow [0,1]$ be a PSA
associated to a density matrix $\rho$.
Let $\{|\alpha_i\rangle\}$  be an orthonormal basis.
Then, the coordinates of $\rho$ in basis $\{|\alpha_i\rangle\langle \alpha_j|\}$
can be obtained from $\Psi$. 
\end{theo}
\begin{proof}
Choosing several contexts $\mathcal{C}$ it is possible (from a theoretical and practical
point of view) to recover $v$ (or $\rho$) by using the theory of quantum tomography, 
(e.g., \cite{altepeter}).\qed\\
\end{proof}

\begin{coro}
The knowledge of a PSA $\Psi$ is equivalent to the knowledge of
the density matrix $\rho_{\Psi}$. In particular, if $\Psi$
is defined by a vector $v_{\Psi}$, $\|v_{\Psi}\|=1$, then
we can recover the vector from $\Psi$.
\end{coro}
\begin{proof}
The first part follows from the previous Theorem. To prove the second part,
apply the previous Theorem to the density matrix $v_{\Psi}v_{\Psi}^t$.
\qed\\
\end{proof}

We are now ready to define density matrices and quantum superpositions within our logos approach.

\begin{definition}
Let $\Psi:\mathcal{G}\rightarrow [0,1]$ be a PSA
and let us choose a complete subgraph $\mathcal{C}
\subseteq\mathcal{G}$ (i.e. an orthonormal basis 
$\{|\alpha_i\rangle\}$). 
The \emph{density matrix} denoted $\rho_{\Psi,{\mathcal C}}$, 
associated to $\Psi$ in basis $\{|\alpha_i\rangle\langle \alpha_j|\}$ is 
denoted as
\[
\rho_{\Psi,{\mathcal C}}=\sum_{ij}^\infty  c_i \overline{c_j} \ |\alpha_i\rangle\langle\alpha_j|.
\]
\end{definition}

\begin{definition}
Let $\mathcal{H}$ be Hilbert space and 
let $\Psi:\mathcal{G}\rightarrow [0,1]$ be a PSA
associated to a vector $v_{\Psi}\in\mathcal{H}$, $\|v_{\Psi}\|=1$.
Let $\mathcal{C}$ be a context given by an orthonormal
basis $\{|\alpha_i\rangle\}$.
A \emph{quantum superposition} (or \emph{quantum situation}) 
relative to $\mathcal{C}$
is given by
\[
QS_{\Psi,\mathcal{C}}=\sum_i c_i|\alpha_i\rangle.
\] 
The dependence of the superposition with respect to $\Psi$
and $\mathcal{C}$ is justified by the previous Theorem.
\end{definition}

It is important to mention that the remarkable properties
of the graph $\mathcal{G}$ can lead to some confusions.
A trivial remark is that $\rho_{\Psi,\mathcal{C}}$ is different from $\Psi$
and clearly different from $\mathcal{C}$ and different from $QS_{\Psi,\mathcal{C}}$.
From the definition of a PSA and the properties of the 
graph $\mathcal{G}$ it is possible to recover the value of $\Psi$ over 
any orthonormal basis and over any
context. Hence, a superposition codifies the values of $\Psi$
over the whole graph $\mathcal{G}$.
\begin{theo}
Let $\mathcal{C}$ be a context and let $QS_{\Psi,\mathcal{C}}$
be a superposition associated to some unknown PSA, $\Psi$.
Then, there exists a unique PSA, $\Psi$, such that
the superposition associated to $\Psi$ over $\mathcal{C}$ is equal to 
$QS_{\Psi,\mathcal{C}}$.
\end{theo}
\begin{proof}
A change of basis from a context $\mathcal{C}$ to a context $\mathcal{C}'$ produces
the values $c_i'$ in a compatible way with Born's rule.
\qed
\end{proof}\\

Summing up, to a PSA $\Psi$ we can assign a multiplicity of different superpositions
\[
QS_{\Psi,\mathcal{C}_1},QS_{\Psi,\mathcal{C}_2},\ldots,QS_{\Psi,\mathcal{C}_n}
\]
one for each maximal context $\{\mathcal{C}_i\}_{i\in I}$. Even more so, 
each superposition determines the others and the whole PSA,
\[
QS_{\Psi,\mathcal{C}_1}\rightsquigarrow
QS_{\Psi,\mathcal{C}_2}\rightsquigarrow 
\ldots\rightsquigarrow
QS_{\Psi,\mathcal{C}_n}\rightsquigarrow
\Psi.
\]

Notice that {\it binary causal powers}, understood as properties of a system that can (or could) become actual in a future instant of time, are incapable of generating the PSA. In conceptual terms, a quantum superposition is defined as a quantum situation; i.e., as an aggregate of commuting immanent powers with their respective potentia. But a binary valuation over an aggregate of commuting immanent powers is not enough to produce the PSA. In order to generate the whole PSA we require the intensive data provided via the Born rule and carried by the potentia of each respective power. This is the reason why the attempts to understand QM in binary terms remarkably fail. The intensive information of the state of affairs described by QM cannot be subsumed by a classical binary model ---something we already know not only from KS theorem but also from Boole-Bell type inequalities \cite{Bell66, Pitowsky94, Svozil17}. This might explain the reason why the notions of `system', `state' and `property' ---grounded on a binary idea of existence--- simply fail when attempting to represent what QM is talking about. QM requires ---as the Born rule indicates explicitly--- an intensive conceptual scheme which is able to account for the statistical nature already present within the orthodox formalism. 

Let us also remark that our logos approach is more general than the Hilbert space formalism. This extension allows us to provide a way of visualizing its elements, adding an {\it anschaulich} content to the formalism through graphs.\footnote{For an interesting analysis of the fruitfulness of diagrams with respect to the representation of mathematical abstract relations, see \cite{Carter17}.} Our logos approach gives a clear insight to the notions compatible with the quantum formalism. In particular, as we have seen above, the notion of `quantum superposition' can be now pictured in a manner which helps us understanding its specific constituents and relations. Taking into account our logos scheme, let us now advance and analyze in more detail some specific characteristics and features of quantum superpositions themselves.

\subsection{Epistemic Contextuality}

The contextual character of quantum superpositions is an aspect of outmost importance when attempting to conceptually represent them. Let us discuss an explicit example in order to clarify these ideas. Consider a typical Stern-Gerlach type experiment where we have produced a PSA, $\Phi$, mathematically represented in the $x$-basis by the ket $|\uparrow_x\rangle$. This can be easily done by filtering the terms $| \downarrow_x\rangle$ of a Stern-Gerlach situation arranged in the $x$-direction. It is a mathematical fact that this ket ---orthodoxly interpreted as `the state of a quantum system'--- can be represented in different bases which diagonalyze a complete set of commuting observables (or immanent powers, in the logos approach). Each basis-dependent representation of the PSA, $\Phi$, is obviously different when considering its physical content. Indeed, it is common to the {\it praxis} of physicists to relate each different basis with a specific measurement context. However, regardless of the widespread claim according to which ``the choice of the context restores classicality''; quantum contexts cannot be related easily to any classical notion without giving up the whole formalism \cite{deRonde16c}. We will come back to this outmost important point in the following subsections. 

As we discussed in detail in \cite{deRonde17a}, the notion of context (or basis), and thus also that of superposition as we defined it above, possesses a physical content which relates a specific subset of epistemic inquiries regarding the abstract PSA, $\Phi$, to a set of {\it Meaningful Operational Statements} (MOS) ---each one of them capable of providing an answer to a specific question (see footnote 6). All MOS, since they imply a specific measurement situation and refer explicitly to the performing act of measuring, are context-dependent. In our example, we know that if we measure  $\Phi$ in the context given by the $x$-basis we will observe with certainty `spin up'. Thus, the MOS related to the ket `$| \uparrow_x\rangle$' is of course: ``if the SG is in the $x$-direction then the result will be `+' with certainty (probability = 1)''. 

In our logos approach, each MOS is related to an immanent power and its specific potentia. The shift from the operational-epistemic discursive level to the conceptual-ontological discursive level makes possible to discuss {\it what is really going on}, independently of specific measurement situations or outcomes. In our Stern-Gerlach example, we can now see that what we have in fact is a particular quantum situation, $QS_{\Phi, x}$, constituted by an immanent power $P_{\uparrow_x}$, $| \uparrow_x\rangle$, with potentia 1 and an immanent power $P_{\downarrow_x}$, $| \downarrow_x\rangle$, with potentia 0.
\begin{center}
\includegraphics[width=4cm]{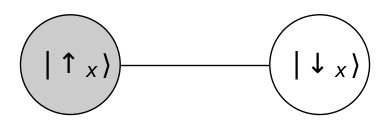}
\end{center}

\noindent Let us remark once again, that even though quantum superpositions are intrinsically contextual, the definition of each immanent power and its potentia is completely non-contextual \cite{deRonde16a, deRondeMassri17a}. 

Going now back to the operational level, physicists are taught that if they want to learn what are the possible outcomes in a different context, for example if they turn the SG to the $y$-direction, then they just need to rewrite the ket $| \uparrow_x\rangle$ in the $y$-basis. Through this change of basis, and according to our previous definition, physicists obtain a different quantum superposition, $QS_{\Phi, y}$:
$$ \frac{1}{\sqrt{2}} | \uparrow_y\rangle +\ \frac{1}{\sqrt{2}} | \downarrow_y\rangle$$

\noindent Writing $| \uparrow_x\rangle$ in the $y$-basis generates a new superposition which implies the analysis of a different section of the graph $\mathcal{G}$. The new superposition also relates itself to the following two MOS. The first one is that ``if the SG is in the $y$-direction then the result will be `+' with probability $1/2$''. The second MOS is: ``if the SG is in the $y$-direction then the result will be `-' with probability $1/2$''. In our logos approach, this subsection of the graph $\mathcal{G}$ relates to the quantum situation where we are considering the powers $| \uparrow_y\rangle$ and $| \downarrow_y\rangle$, both with 
potentia $1/2$, 
\begin{center}
\includegraphics[width=4cm]{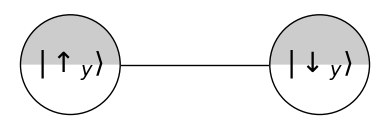}
\end{center}

The same will happen with any $i$-context of inquiry (determined by a particular $i$-basis), each one of them will be related to a particular quantum superposition, $QS_{\Phi, i} = c_{i1} | \uparrow_i\rangle + \ c_{i2} | \downarrow_i\rangle$, a specific set of MOS in the epistemic level of analysis and a specific subset of powers, \{$P_i$\}, and potentia, \{$p_i$\}, in the ontological level of discourse. Thus, we have provided an explicit bridge between the MOS derived from the mathematical formalism and a conceptual representation. In this way we are able to connect the operational epistemic level of inquiry to the conceptual ontic level of objective representation \cite{deRonde16a}. Our conceptual representation in terms of intensive immanent powers is also capable to escape KS contextuality, generating a coherent counterfactual objective discourse which includes all the MOS contained within the theory. 

In order to be even more clear regarding our explanation, and before concluding, let us recall from \cite{deRonde16c} the following two very helpful definitions:\\

\noindent {\it {\sc Epistemic Incompatibility of Measurements:} Two contexts are epistemically incompatible if their measurements cannot be performed simultaneously.}\\

\noindent {\it  {\sc Ontic Incompatibility of Existents:} Two contexts are ontically incompatible if their formal elements cannot be considered as simultaneously preexistent.}\\

What we have seen in this subsection is that even though a quantum situation (or context) $X$ can be {\it epistemically incompatible} with another quantum situation (or context) $Y$, all intensive immanent powers are {\it ontologically compatible}. This is a direct consequence of the fact we can define each immanent power and its respective potentia  ---through the Born rule and the Graph--- completely independently of any context. Contexts play no role whatsoever in the definition of what constitutes physical reality. While ---due to KS conetxtuality and the impossibility of GBV--- projection operators interpreted as properties are ontically incompatible in the quantum Hilbert formalism, intensive immanent powers ---escaping KS contextuality through GIV--- are ontically compatible. This means that the contextuality present in quantum situations ---understood as aggregates of immanent powers--- is purely epistemic, not ontological. Everything we have explained is very easily visualized in our logos approach, thus adding an {\it anschaulich} content to the theory. While quantum situations expose the epistemic incompatibility of measurements, immanent powers remain always ontically compatible.

By providing an invariant global account of $\Psi:\mathcal{G}\to[0,1]$, we have restored the idea according to which {\it the representation provided by a closed physical theory provides an objective account of physical reality}. The quantum reality ``behind the veil'' ---to use an expression from Bernard D'Espagnat (see e.g., \cite{DEspagnat03})--- is fully expressed by the formal representation provided in the PSA.\footnote{This obviously does not imply the naive realist claim that ---consequently--- we are finally representing {\it reality as it is.} See for a discussion of this important point: \cite{deRonde16b}.} Accordingly, measurements only play an epistemic role; subjective choices ---imposed by the complementarity principle--- are not required in order to define the state of affairs. This resolution to the quantum riddle goes obviously against Bohr's famous claim that the most important epistemological lesson to be learnt from QM is that, {\it we are not only spectators but also actors in the great drama of existence.} Our proposal invites everyone to recall the fact that physical theories are ---since their origin--- attempts to provide (ontological) representations of {\it physis}. And that we can only understand ourselves, (empirical) subjects, as being part of a representation.\footnote{Wittgenstein's assertion that ``the limits of my language mean the limits of my world'', captures the fact we are always speaking from {\it within} a particular representation.} We, subjects, are not the center of the Universe. We (empirical) subjects, are not an empire within an empire. So please, let us invite everyone ---including Bohr--- to just sit, relax and enjoy the quantum drama pictured objectively in the following terms.
\begin{center}
\includegraphics[width=14em]{csg-pac.png}
\end{center}
                 
The price we have willingly paid in order to restore an objective representation ---in which, contrary to the Bohrian orthodoxy, measurement plays only an epistemic role and the (subjective) choice of the context is not required in order to make reference to the state of affairs--- is to consider a potential realm of existence which is completely independent of the actual realm. Such potential realm exists in this world alone and is observable according to the specific {\it conditions of observability} provided by immanent intensive powers themselves. 

As we shall discuss in the following subsections, there are two main reasons why quantum superpositions cannot be related to an objective representation grounded on the actual (binary) mode of existence.\footnote{This is of course in case we do not want to overpopulate reality with unobservable worlds in order to explain a `pointer reading' we did not observe, or shift completely the focus to observability itself and end up discussing the problem of human consciousness.} The first one, relates to the KS-contextual character of binary valuations themselves; which precludes the possibility of interpreting each ket as a definite valued property. The second, deals with the paraconsistent character of superposed Schr\"odinger cat states \cite{daCostadeRonde13} and the impossibility to interpret the probabilities of outcomes in terms of ignorance about an actual state of affairs \cite{deRonde16a, deRonde17a}. Since we already addressed the problem of KS contextuality in \cite{deRondeMassri17a} let us now concentrate in the latter reasons related to paraconsistency and the interpretation of probability.

\subsection{Potential Paraconsistency}

In \cite{daCostadeRonde13} one of the authors of this paper together with Newton da Costa, argued in favor of considering quantum superpositions in terms of a paraconsistent logic. The idea came from the acknowledgment that Schr\"odinger cat states, composed by contradictory properties ---such as `atom decayed' and `atom not decayed'---, seem to constitute nonetheless an individual quantum existent. The intuition is quite straightforward. As explained by Valia Allori in \cite[p. 175]{Allori16}:  ``many have thought that the real lesson of quantum mechanics is that the dream of the scientific realist is impossible, since quantum mechanics has been taken to suggest that physical objects have contradictory properties, like being in a place and not being in that place at the same time, or that properties do not exist at all independently of observation.'' Of course, these ideas were already implicit in the famous 1935 ``cat paper'' in which Schr\"odinger discussed the absurdity of using the classical notion of `cat' in order to interpret quantum superpositions \cite{Schr35}. Elsewhere \cite{deRonde17a}, we have argued that this paper can be read as an {\it ad absurdum} proof of the untenability of the classical notion of entity (object or system) to account for a coherent interpretation of quantum superpositions. Indeed, the very categorical precondition of the physical notion of `system' (`entity' or `object') is that any existent of this type is constituted by definite valued {\it non-contradictory} properties. 

The mentioned paper, \cite{daCostadeRonde13}, inaugurated a debate about the possibility to consider a paraconsistent logic for QM (see \cite{ArenhartKrause14a, ArenhartKrause15, ArenhartKrause16, daCostadeRonde15, daCostadeRonde16, deRonde15a, deRonde17e, Eva16, KrauseArenhart15}). In particular, D\'ecio Krasue and Jonas Arenhart have raised arguments against the idea that quantum superpositions should be understood in terms of {\it contradictory} properties. Instead, they claim that the notion of {\it contrariety} is much better suited to discus about superpositions in quantum theory. In a recent reply to their criticisms, one of us has argued that the confusion comes from imposing ---implicitly--- an actualist metaphysical representation within the discussion about quantum superpositions and their measurement \cite{deRonde17e}. Since we have argued that superpositions provide a formal account in terms of an ontological potential realm, the contradiction we are discussing about should be obviously understood as making reference to that same realm. This means we've been discussing about a {\it potential contradictions}, not about {\it actual contradictions} ---as Krause and Arenhart presuppose implicitly in their analysis. In order to be more rigorous, let us recall the traditional definitions of the famous Aristotelian square of opposition and the specific meaning of both {\it contradiction} and {\it contrariety}.

\begin{enumerate}
\item[]{\bf Contradiction Propositions:} $\alpha$ and $\beta$ are {\it contradictory} when both cannot be true and both cannot be false.
\item[]{\bf Contrariety Propositions:} $\alpha$ and $\beta$ are {\it contrary} when both cannot be true, but both can be false.
\item[]{\bf Subcontrariety Propositions:} $\alpha$ and $\beta$ are {\it subcontraries} when both can be true, but both cannot be false.
\item[]{\bf Subaltern Propositions:} $\alpha$ is {\it subaltern} to proposition $\beta$ if the truth of $\beta$ implies the truth of $\alpha$.
\end{enumerate}
It is easy to see that the notion of contradiction, in the case we want to advance in the possibility to understand quantum superpositions in terms of a PSA, requires the specification of the realm that is being presupposed within the discussion. A more accurate definition would be the following.

\begin{enumerate}
\item[]{\bf Potential Contradiction Propositions:} $\alpha$ and $\beta$ are {\it potential contradictory} when both cannot be true and both cannot be false {\it in the actual realm}.
\end{enumerate}

An intuitive example can be provided in order to further understand this definition. As we discussed in \cite{deRonde17e}, consider the immanent power possessed by Messi of `shooting penalties'. The power of shooting a penalty comprises two contradictory actual effectuations, namely, `to score a goal from the penalty' and also `to fail the penalty'. Since both powers preexist in the potential realm, but both cannot be true and both cannot be false in the actual realm, we call them {\it potential contradictions}. This is also an accurate account of what is going on with quantum superpositions of the type discussed in the previous section. Indeed, a quantum superposition of the type $ c_{1} | \uparrow \rangle + \ c_{2} | \downarrow_i \rangle$,  can be related to two MOS which even though existent in the potential realm, cannot be, in the actual realm, both true or both false. Instead, we always observe only one (true) possibility actualized in each measurement. Following this definition, the logos approach goes in line with the idea that quantum superpositions are composed of contradictory powers; i.e. powers which, about one and the same phenomena, simultaneously refer to contradictory results in the potential realm but produce a single true result in the actual realm. In fact, the logic arising from our logos approach is paraconsistent, 
see \cite{Eva16} and \cite[\S 3.3]{Angot15}.

\subsection{Objective Probability (and Epistemic Measurements)}

Gleason's theorem can be understood, following Karl Svozil \cite{Svozil17}, as precluding the possibility of two valued measures in quantum logic.\footnote{This is completely analogous to the interpretation of {\it quantum possibility} discussed in \cite{RFD14, DFR06}.} In our own terms, this can be also understood as the impossibility of relating the quantum probability measure to an actual state of affairs.\footnote{Let us remark that when considering the problem of representing quantum physical reality the application of the Bayesian interpretation of {\it subjective probability} misses completely the point. An ontological question cannot be addressed from an epistemological perspective. QBism does use the Bayesian subjectivist interpretation of probability, but at the cost of denying the {\it reference} of QM to physical reality itself \cite{FuchsPeres00, QBism13}.} In turn, this impossibility to apply a classical ignorance interpretation is obviously applicable to quantum superpositions themselves. It is well known that one cannot apply an ignorance interpretation to the terms of a superposition in case one is also willing to respect the orthodox Hilbert formalism. But it is interesting that this fact, that $QS$ cannot be understood as the state of a system composed by definite valued properties, can be also derived naturally as a consequence of KS theorem itself.\footnote{This result is just an expression of the conclusion derived by Schr\"odinger in \cite[p. 153]{Schr35} regarding the notion of {\it state} in QM: ``The classical notion of {\it state} becomes lost [in QM], in that at most {\it half} of a complete set of variables can be assigned definite numerical values''.}   
\begin{coro}[KS for Superpositions]
Let $\Psi:\mathcal{G}\to[0,1]$ be a PSA, $\mathcal{C}=\{P_{\alpha_i}\}$ a context
and $QS_{\Psi,\mathcal{C}}=\sum c_i|\alpha_i\rangle$ a superposition.
Then, a binary valuation for the elements $\{|\alpha_i\rangle\}$ pertaining to the superposition $QS_{\Psi,\mathcal{C}}$ that preserves the compatibility conditions of $\Psi$ is precluded.
\end{coro}
\begin{proof}
Assigning true/false values to each $|\alpha_i\rangle$ leads to a contradiction as it is proven in KS Theorem. This shows that one cannot interpret in terms of ignorance the terms pertaining to a quantum superposition  $QS_{\Psi,\mathcal{C}}$ without confronting the whole structure provided by $\Psi$.
\qed\\
\end{proof}

It is not possible to provide an ignorance interpretation of the terms pertaining to a quantum superposition. Furthermore, each term can be related to a MOS; i.e., every term is predicted by the theory. Since every term bears some possibility of being observed. And since the theory is empirically adequate; it has been tested in many, many occasions and we have always found that experimental observations agree systematically with the predictions of the theory. That can only mean ---from a realist perspective, at least--- that all terms, according to the theory, bear some kind of existence. This intuitive idea ---also present in the many worlds interpretation of QM--- according to which, if a theoretical formalism talks about elements that can be predicted operationally, then there must exist some {\it reference} of these elements to reality, is at the basis of physical realism itself.\footnote{The discussion about what exactly is this relation exceeds the scope of the present paper. We leave this discussion for a future work.} 

As we discussed extensively in \cite{deRonde16a, deRondeMassri17a}, in case we are willing to give up our actualist metaphysical prejudices there is a possible solution to the riddle. If we replace the classical account by one in terms of immanent powers with definite potentia, we can indeed restore not only an objective representation of what QM is talking about, but also the epistemic character of quantum measurements themselves. This requires to take the formalism very seriously, which means for us, not to believe in a multiplicity of unobservable worlds nor to shift the analysis to our own consciousness, but rather to consider the invariant conditions implied by the mathematical formalism itself. Instead of going from a presupposed (classical) metaphysical scheme into the formalism we go from the mathematical formalism into (non-classical) metaphysics. And if we look closely to the formalism, it is of course the Born rule which provides an {\it invariant definition} of `value' that is completely independent of the choice of any context \cite{deRonde16a, deRondeMassri16}.\\ 

\noindent {\it
{\bf Born Rule:} Given a vector $\Psi$ in a Hilbert space, the following rule allows us to predict the average value of (any) observable $P$. 
$$\langle \Psi| P | \Psi \rangle = \langle P \rangle$$
As the calculation makes explicit by itself, this definition is independent of the choice of any particular basis.}\\

\noindent Taken seriously, this line of thought naturally leads us to the consideration of a {\it generalized element of physical reality} \cite{deRonde16a}.\\

\noindent {\it {\bf Generalized Element of Physical Reality:} If we can predict in any way (i.e., both probabilistically or with certainty) the value of a physical quantity, then there exists an element of reality corresponding to that quantity.}\\

Probability has in this case also an objective character since it provides accurate knowledge of intensive existents. Indeed, within our logos approach, empirical subjects play only a passive role. We are now ready to argue ---against the Bohrian orthodoxy--- that in QM, as in classical physics, empirical subjects are only spectators and never actors when observing the dramatic evolution of the quantum (potential) state of affairs. According to the logos account of QM, empirical subjects ---like Einstein wanted--- are completely ``detached'' from the representation of quantum physical reality. So just like the moon has a definite position regardless of we observing it or not, an immanent power has a definite potentia whether we choose to measure it or not. 

Now, since some quantum superpositions also discuss {\it potentially contradictory} MOS, it then follows that $QS_{\Psi,\mathcal{C}}$ can be composed of contradictory existents ---such as the superposition discussed in the previous Stern-Gerlach example. 

\begin{theo}[Superpositions are not `States of a System']
A quantum superposition cannot be understood as the state of a system; i.e. a particular set of definite valued non-contradictory properties.
\end{theo}
\begin{proof}
From the previous result, it follows that a superposition possesses contradictory elements which preexist to measurements (i.e., they cannot be interpreted in terms of ignorance).
\qed\\
\end{proof}

\noindent Since the components of a quantum superposition can be contradictory ---in the sense discussed above---, the interpretation in terms of a (non-contradictory) `system' (or object) is precluded. The common escape route of researchers of QM using these classical notions is to recall that this is ``quantum'' and thus ``weird''. But this of course does not provide any justification ---beyond the dogmatic belief that QM talks about ``small particles''--- of why the notions of `system' and `state' have been used within the semantics applied to the mathematical formalism.

\section*{Conclusions}

In this paper we have provided a characterization of quantum superpositions within our logos categorical approach to QM. Our representation in both conceptual terms, through the notions of immanent power and potentia, and in formal terms, through the introduction of Graphs, allowed us to provide an {\it anschaulich} content to the theory which gives an intuitive explanation of many of the main features already present in the orthodox Hilbert formalism. In particular, we discussed an epistemic account of contextuality which --escaping KS theorem--- presents an objective representation of the state of affairs the theory talks about. We also provided an intuitive understanding of the paraconsistent character within typical Stern-Gerlach type experiments. Our approach also allows to provide a clear coherent distinction between different mathematical levels ---overseen in the literature due to the semantical interpretation in terms of `systems', `states' and `properties'. Finally, we have provided an interpretation of quantum entanglement which allows both a graphical and conceptual understanding that goes beyond the operational reference to measurement outcomes.

\section*{Acknowledgements} 

This work was partially supported by the following grants: FWO project G.0405.08 and FWO-research community W0.030.06. CONICET RES. 4541-12 and the Project PIO-CONICET-UNAJ (15520150100008CO) ``Quantum Superpositions in Quantum Information Processing''.


\begin{thebibliography}{1}

\bibitem{quasitopoi}  Ad{\'a}mek, J. and Herrlich, H., 1986, ``Cartesian closed categories, quasitopoi and topological universes'', {\em Comments on Mathematics University of Carolina}, {\bf 27}, 235-257.

\bibitem{Allori16} Allori, V., 2016, ``Primitive Ontology and the Classical World'', in {\it Quantum Structural Studies}, pp. 175-199, R. E. Kastner, J. Jeknic-Dugic and G. Jaroszkiewicz (Eds.), World Scientific, Singapore.

\bibitem{altepeter} Altepeter J.B., James D.F., Kwiat P.G. 2004, ``4 Qubit Quantum State Tomography'', in {\it Quantum State Estimation}, M. Paris and J. Rehácek (Eds.), Springer, Berlin.

\bibitem{Angot15} Angot-Pellissier, R., 2015, ``The Relation Between Logic, Set Theory and Topos Theory as It Is Used by Alain Badiou''  in {\it The Road to Universal Logic}, pp. 181-200, A. Koslow, A. Buchsbaum (Eds.), Springer, Switzerland. 

\bibitem{ArenhartKrause14a} Arenhart, J. R. and Krause, D., 2014, ``Oppositions in Quantum Mechanics",  in {\it New dimensions of the square of opposition}, Jean-Yves B\'eziau and Katarzyna Gan-Krzywoszynska (Eds.), 337-356, Philosophia Verlag, Munich.

\bibitem{ArenhartKrause15} Arenhart, J. R. and Krause, D., 2015, ``Potentiality and Contradiction in Quantum Mechanics.", in {\it The Road to Universal Logic (volume II)}, pp. 201-211, Arnold Koslow and Arthur Buchsbaum (Eds.), Birkh\"auser, Cham. 

\bibitem{ArenhartKrause16} Arenhart, J. R. and Krause, D., 2016, ``Contradiction, Quantum Mechanics, and the Square of Opposition'', {\it Logique et Analyse}, {\bf 59}, 273-281. 

\bibitem{Bell66} Bell, J., 1966, ``On the Problem of Hidden Variables in Quantum Mechanics'', {\it Review of Modern Physics}, {\bf 38}, 477.

\bibitem{Bohr35} Bohr, N., 1935, ``Can Quantum Mechanical Description of Physical Reality be Considered Complete?'', {\it Physical Review}, {\bf 48}, 696-702.

\bibitem{Butterfield17} Butterfild, J., {\it Interview: Jeremy Butterfield: What is contextuality?}, (published 3rd of March, 2017). URL: https://www.youtube.com/watch?v=ZJAnixX8T4U

\bibitem{cabello} Cabello, A., Estebaranz, J.M. and Garc{\'{\i}}a-Alcaine, G., 1996, ``Bell-{K}ochen-{S}pecker theorem: a proof with {$18$} vectors'', {\em Physics Letters A}, {\bf 4}, 183-187.

\bibitem{Carter17} Carter, J., 2017, ``Exploring the fruitfulness of diagrams in mathematics'', {\it Synthese}, DOI 10.1007/s11229-017-1635-1. (philsci-archive:14130)

\bibitem{daCostadeRonde13} da Costa, N. and de Ronde, C., 2013, ``The Paraconsistent Logic of Quantum Superpositions'', {\it Foundations of Physics}, {\bf 43}, 845-858.

\bibitem{daCostadeRonde15} da Costa, N. and de Ronde, C., 2015, ``The Paraconsistent Approach to Quantum Superpositions Reloaded: Formalizing Contradictory Powers in the Potential Realm'', preprint. (quant-ph/arXive:1507.02706)

\bibitem{daCostadeRonde16} da Costa, N. and de Ronde, C., 2016, ``Revisiting the Applicability of
Metaphysical Identity in Quantum Mechanics'', preprint. (quant-ph/arXive:1609.05361)

\bibitem{deRonde15a} de Ronde, C., 2015, ``Modality,  Potentiality and Contradiction in Quantum Mechanics'', {\it New Directions in Paraconsistent Logic}, J.-Y. Beziau, M. Chakraborty and S. Dutta (eds), Springer, Berlin.

\bibitem{deRonde16a} de Ronde, C., 2016, ``Probabilistic Knowledge as Objective Knowledge in Quantum Mechanics: Potential Immanent Powers instead of Actual Properties'', in {\it Probing the Meaning of Quantum Mechanics: Superpositions, Semantics, Dynamics and Identity}, pp. 141-178, D. Aerts, C. de Ronde, H. Freytes and R. Giuntini (Eds.), World Scientific, Singapore.

\bibitem{deRonde16b} de Ronde, C., 2016, ``Representational Realism, Closed Theories and the Quantum to Classical Limit'', in {\it Quantum Structural Studies}, pp. 105-135, R. E. Kastner, J. Jeknic-Dugic and G. Jaroszkiewicz (Eds.), World Scientific, Singapore.

\bibitem{deRonde16c} de Ronde, C., 2016, ``Unscrambling the Omelette of Quantum Contextuality: Preexistent Properties or Measurement Outcomes?'', preprint. (quant-ph/arXive:1606.03967)

\bibitem{deRonde17a} de Ronde, C., 2017, ``Quantum Superpositions and the Representation of Physical Reality Beyond Measurement Outcomes and Mathematical Structures'', {\it Foundations of Science}, https://doi.org/10.1007/s10699-017-9541-z. (quant-ph/arXive:1603.06112)

\bibitem{deRonde17c} de Ronde, C., 2017, ``Immanent Powers versus Causal Powers (Propensities, Latencies and Dispositions) in Quantum Mechanics'', in {\it Probing the Meaning of Quantum Mechanics}, D. Aerts, M.L. Dalla Chiara, C. de Ronde and D. Krause (Eds.), World Scientific, Singapore, forthcoming. (quant-ph/arXive:1711.02997) 

\bibitem{deRonde17d} de Ronde, C., 2017, ``Causality and the Modeling of the Measurement Process in Quantum Theory'', {\it Disputatio}, forthcoming. (quant-ph/arXive:1310.4534)

\bibitem{deRonde17e} de Ronde, C., 2017, ``A Defense of the Paraconsistent Approach to Quantum Superpositions (Reply to Arenhart and Krause)'', {\it Metatheoria}, forthcoming.

\bibitem{deRonde17b} de Ronde, C., 2017, ``Potential Truth in Quantum Mechanics'', preprint. 

\bibitem{deRondeMassri16} de Ronde, C. and Massri, C., 2017, ``Kochen-Specker Theorem, Physical Invariance and Quantum Individuality'', {\it Cadernos da Filosofia da Ciencia}, forthcoming. (quant-ph/arXive:1412.2701)

\bibitem{deRondeMassri17a} de Ronde, C. and Massri, C., 2018, ``The Logos Categorical Approach to Quantum Mechanics: I. Kochen-Specker Contextuality and Global Intensive Valuations.'', preprint. (quant-ph/arXive:1801.00446) 

\bibitem{RFD14} de Ronde, C., Freytes, H. and Domenech, G., 2014, ``Interpreting the Modal Kochen-Specker Theorem: Possibility and Many Worlds in Quantum Mechanics'', {\it Studies in History and Philosophy of Modern Physics}, {\bf 45}, 11-18.

\bibitem{DEspagnat03} D'Espagnat, B., 2003, {\it Veiled Reality: An Analysis of Present-Day Quantum Mechanical Concepts}, Westview Press, Colorado. 

\bibitem{Deutsch97} Deutsch, D., 1997, {\it The Fabric of Reality}, Penguin, London.

\bibitem{graphtheory}
Diestel, R., 2010,
\newblock {\em Graph theory}, \newblock Springer, Heidelberg.

\bibitem{DFR06} Domenech, G., Freytes, H. and de Ronde, C., 2006,
``Scopes and limits of modality in quantum mechanics",
\textit{Annalen der Physik}, {\bf 15}, 853-860.

\bibitem{Eva16} Eva, B., 2016, ``A Topos Theoretic Framework for Paraconsistent Quantum Theory'', {\it Probing the Meaning of Quantum Mechanics: Superpositions, Dynamics, Semantics and Identity},  pp. 340-350, D. Aerts, C. de Ronde, H. Freytes and R. Giuntini (Eds.), World Scientific, Singapore.

\bibitem{FuchsPeres00} Fuchs, C.A., and Peres, A., 2000, ``Quantum theory needs no `interpretation' '', {\it Physics Today}, {\bf 53}, 70.

\bibitem{QBism13} Fuchs, C.A., Mermin, N.D. and Schack, R., 2014, ``An introduction to QBism with an application to the locality of quantum mechanics'', {\it American Journal of Physics}, {\bf 82}, 749. (quant-ph/arXive:1311.5253)

\bibitem{Griffiths13} Griffiths, R. B., 2013, ``Hilbert space quantum mechanics is non contextual'',  {\it Studies in History and Philosophy of Modern Physics}, {\bf 44}, 174-181.

\bibitem{HawkingPenrose} Hawking, S. and Penrose, R., 2015, {\it The Nature of Space and Time}, Princeton University Press, Princeton. 

\bibitem{KrauseArenhart15} Krause, D. and Arenhart, J., 2015, ``A Logical account of Superpositions'', in {\it Probing the Meaning and Structure of Quantum Mechanics: Superpositions, Semantics, Dynamics and Identity}, pp. 44-59, D. Aerts, C. de Ronde, H. Freytes and R. Giuntini (Eds.), World Scientific, Singapore.

\bibitem{maclane98} Mac~Lane, S., 1998, {\it Categories for the Working Mathematician}, {\it  Graduate Texts in Mathematics, volume~5}, Springer-Verlag, New York.

\bibitem{Penrose17} Penrose, R., Interview series: ``Many Worlds of Quantum Theory'',  in {\it Closer to Truth}. URL: https://www.closertotruth.com/series/many-worlds-quantum-theory

\bibitem{Pitowsky94} Pitowsky, I., 1994, ``George Boole's
`Conditions of Possible Experience' and the Quantum Puzzle", {\it
The British Journal for the Philosophy of Science}, {\bf 45},
95-125.

\bibitem{Schr35} Schr\"odinger, E., 1935, ``The Present Situation in Quantum Mechanics'', {\it Naturwiss}, {\bf 23}, 807-812. Translated to english in {\it Quantum Theory and Measurement}, J. A. Wheeler and W. H. Zurek (Eds.), 1983, Princeton University Press, Princeton.

\bibitem{Steane03} Steane, A.M., 2003, ``A Quantum Computer Only Needs One Universe'', {\it Studies in History and Philosophy of Modern Physics, B}, {\bf 34}, 469-478 (quant-ph/arXive:0003084v1).

\bibitem{Svozil17} Svozil, K., 2017, ``Classical versus quantum probabilities and correlations'', preprint (quant-ph/arXive:1707.08915).

\bibitem{Wallace12} Wallace, D., 2012, {\it The Emergent Multiverse: Quantum Theory according to the Everett. Interpretation}, Oxford University Press, Oxford.

\bibitem{Wallace17} Wallace, D., Interview series: ``Many Worlds of Quantum Theory'',  in {\it Closer to Truth}. URL: https://www.closertotruth.com/series/many-worlds-quantum-theory

\bibitem{weyl49} Weyl, H., 1949, {\it Philosophy of Mathematics and Natural Science} (Revised
  and augmented English edition based on a translation by Olaf Helmer), Princeton University Press, Princeton.

\bibitem{WZ} Wheeler, J. and Zurek, W. (Eds.) 1983, {\it Theory and Measurement}, Princeton University Press, Princeton.

\end{thebibliography}

\end{document}